\documentclass[conference]{IEEEtran}
 
\usepackage{aliascnt}
\usepackage{cite}
\usepackage{wrapfig}
\usepackage{gensymb}

\IEEEoverridecommandlockouts
\newenvironment{talign}
 {\align}
 {\endalign}
\usepackage{multirow}
\usepackage{enumitem}
\usepackage{pifont}
\usepackage{amsmath,amssymb,amsfonts,amsthm}
\usepackage{algorithmic}
\usepackage[ruled,linesnumbered]{algorithm2e}
\usepackage{graphicx}
\usepackage{textcomp}
\usepackage{caption}
\usepackage[hidelinks]{hyperref}
\usepackage{float}
\usepackage[utf8]{inputenc}
\usepackage{tabularx}
\usepackage{fancyhdr}
\usepackage{booktabs}
\usepackage{cases}
\usepackage{subfigure}
\usepackage{bm}
\def\BibTeX{{\rm B\kern-.05em{\sc i\kern-.025em b}\kern-.08em
    T\kern-.1667em\lower.7ex\hbox{E}\kern-.125emX}}

\usepackage{graphicx}
\usepackage{mathrsfs}
\usepackage{bm}

\DeclareRobustCommand{\vect}[1]{\bm{#1}}
\newtheoremstyle{slanted}
{0em plus 0em minus 0em}%   Space above
  {0em plus 0em minus 0em}%   Space below
  {\em}%  Body font
  {}%          Indent amount (empty = no indent, \parindent = para indent)
  {\bfseries}% Thm head font
  {.}%         Punctuation after thm head
  { }%     Space after thm head: " " = normal interword space;
  {}%          Thm head spec (can be left empty, meaning `normal')

\theoremstyle{slanted}

\theoremstyle{slanted}

\theoremstyle{slanted}
\newtheorem{theorem}{Theorem}

\theoremstyle{slanted}

\theoremstyle{slanted}

\theoremstyle{slanted}

\theoremstyle{slanted}

\usepackage{enumitem}

\begin{document}

\title{\fontsize{22}{26}\selectfont User Association and Resource Allocation in Large Language Model Based Mobile Edge Computing System over 6G Wireless Communications}
% \author{Liangxin~Qian,~\IEEEmembership{Graduate Student Member, IEEE},
%         Jun~Zhao,~\IEEEmembership{Member, IEEE}\vspace{-25pt}
% \thanks{L. Qian and J. Zhao are with the School of Computer Science and Engineering at Nanyang Technological University (NTU), Singapore. (e-mail: qian0080@e.ntu.edu.sg and junzhao@ntu.edu.sg).}
% }

\author{\IEEEauthorblockN{Liangxin Qian and Jun Zhao\\}
\IEEEauthorblockA{{School of Computer Science and Engineering} \\
{Nanyang Technological University}\\
Singapore\\
qian0080@e.ntu.edu.sg, junzhao@ntu.edu.sg\vspace{-25pt}}

}

% \maketitle
% \pagestyle{plain}
% \thispagestyle{plain}
\maketitle
\thispagestyle{fancy}
\pagestyle{fancy}
\lhead{This paper appears in the 2024 IEEE 99th Vehicular Technology Conference (VTC).}
\cfoot{\thepage}
\renewcommand{\headrulewidth}{0.4pt}
\renewcommand{\footrulewidth}{0pt}

\begin{abstract}
In the rapidly evolving landscape of large language models (LLMs) and mobile edge computing for 6G, the need for efficient service delivery to mobile users with constrained computational resources has become paramount. Addressing this, our paper delves into a collaborative framework for model training where user data and model adapters are shared with servers to optimize performance. Within this framework, users initially update the first several layers of the adapters while freezing the other layers of them, leveraging their local datasets. Once this step is complete, these partially trained parameters are transmitted to servers. The servers, equipped with more robust computational capabilities, then update the subsequent layers. After this training, they send the enhanced parameters back to the users. This collaborative training approach ensures that mobile users with limited computational capacities can still benefit from advanced LLM services without being burdened by exhaustive computations. Central to our methodology is the DASHF algorithm, which encapsulates the \underline{D}inkelbach algorithm, \underline{a}lternating optimization, \underline{s}emidefinite relaxation (SDR), the \underline{H}ungarian method, and a pioneering \underline{f}ractional programming technique from a recent IEEE JSAC paper~\cite{zhao2023human}. The crux of DASHF is its capability to reformulate an optimization problem as Quadratically Constrained Quadratic Programming (QCQP) via meticulously crafted transformations, making it solvable by SDR and the Hungarian algorithm. Through extensive simulations, we demonstrate the effectiveness of the DASHF algorithm, offering significant insights for the advancement of collaborative LLM service deployments.
\end{abstract}


\section{Introduction}
The proliferation of large language models (LLMs) marks a monumental leap in the realms of artificial intelligence and natural language processing. These models, with their deep structures and vast parameter sizes, offer capabilities that redefine the benchmarks of machine-human interactions for 6G~\cite{gao2023llama}. However, the very nature of their size and intricacy means they cannot be effortlessly deployed, especially in constrained environments like mobile devices~\cite{zhang2023llama}.

Mobile edge computing (MEC) environments, designed to bring computation closer to the data source, seem like a perfect fit for deploying LLMs. Still, they present their own set of challenges. Mobile devices are constrained by computational resources and battery life, making it strenuous to run these heavyweight LLMs efficiently~\cite{dong2023lambo}. Additionally, the unpredictability of wireless communication, with its fluctuating data rates and potential for high latency, complicates the seamless integration of LLMs~\cite{shen2023large}.
\subsection{Related Work}
In this section, we briefly introduce some recent research on the combination of LLM and MEC.

Recent advancements in LLMs and foundation models have markedly propelled natural language processing and artificial intelligence (AI) tasks, with MEC playing a crucial role in reducing latency and boosting efficiency. Among these, the computing and network convergence (CNC) brain \cite{hong2023intelligence} stands out, employing AI to effectively match resources in computing and network convergence environments, showcasing a trend towards more autonomous, self-optimizing networks suitable for Internet of things (IoT) and advanced smart applications.

Simultaneously, in the domain of federated learning, the introduction of PromptFL \cite{guo2023promptfl} represents a significant leap. This method revolutionizes federated learning by adopting prompt training in environments with limited resources, thereby improving processing efficiency and ensuring better resource distribution in distributed computing setups. PromptFL's emergence underlines the critical need for AI systems that are both adaptable and scalable, especially in dynamic network environments.

Additionally, the fusion of Generative AI with mobile edge networks in the sphere of 6G communications, as discussed in \cite{lai2023resource}, opens new avenues for enhancing network intelligence and efficiency, especially in scenarios with limited resources. This integration hints at future developments in generative mobile edge networks, which could lead to more robust, intelligent, and self-learning network systems, significantly boosting user experiences and operational efficiency in a variety of sectors.
\subsection{Challenges and motivations} To meet the growing demand for on-the-fly LLM services, there's a pressing need to address these issues. This involves optimizing the LLMs for constrained devices and innovating on the wireless communication front. A potential solution lies in a collaborative approach: a synergy where local computations on mobile devices are harmoniously complemented by offloading specific, intensive tasks to more capable servers. Such a paradigm can make the promise of LLMs in MEC environments a tangible reality. 

At present, there are few or no papers on user connection and communication computing resource allocation using MEC for assisting network training in LLM scenarios.

\subsection{Studied problem}
In this paper, we explore the LLM-driven MEC system and introduce the novel concept of the user service-cost ratio (SCR), represented as $\frac{\text{service score}}{\text{cost}}$. This metric eloquently captures the balance between user service scores and the crucial factors of delay and energy consumption within mobile computing environments. The user service score amalgamates a user's wireless and computational resources as perceived by the server. Cost consumption embodies the cumulative delay and energy expenditure of both users and servers. Given the computational challenges posed by LLMs, we hypothesize that users begin by training the initial layers of the adapters with their local data. Once this preliminary training is completed, users send these trained parameters to the servers. Servers, equipped with more computational resources, take over from this point, training the subsequent layers of the adapters. Specifically, servers then assign both wireless and computational resources to each user. This includes bandwidth, user and server transmission power, and the GPU computing resources of users and servers. Upon completion, these refined parameters are then relayed back to the users by servers.
\subsection{Main contributions}
To the best of our knowledge, our paper is the first to explore user association and resource allocation in the LLM wireless communication scenario. Our contributions include a novel joint optimization problem, the concept of SCR, and a novel alternating optimization algorithm as follows:

\noindent $\bullet$ Joint Optimization of User-Sercer Adapter Parameter Training Ratio and User Association: We propose a joint optimization problem that optimizes user adapter training offloading and user association for tailored LLM service to users.

\noindent $\bullet$ Introduction of the User Service-Cost Ratio: The concept of the user service-cost ratio (SCR) is introduced. SCR quantifies the balance between user service scores and the overall delay and energy consumption in the entire uplink and downlink communication. It provides a valuable metric for assessing the trade-off between the user's obtained service resources and resource efficiency.

\noindent $\bullet$ Innovative Alternating Optimization Approach: We propose an innovative Alternating Optimization (AO) approach called DASHF, which represents the combination of the \underline{D}inkelbach algorithm, \underline{a}lternating optimization, \underline{s}emidefinite relaxation (SDR), the \underline{H}ungarian algorithm, and a novel \underline{f}ractional programming (FP) technique by~\cite{zhao2023human}  published in IEEE JSAC recently. The most challenging part of DASHF is to rewrite an optimization problem as Quadratically Constrained Quadratic Programming (QCQP) via carefully constructed transformations, in order to leverage SDR and the Hungarian algorithm to obtain a solution. Initially, it addresses the optimization of user connections and adapter parameter training ratios as a single QCQP problem. Subsequently, it delves into the optimization of communication-computation resource allocation (for bandwidth, transmit power of users and servers, and computing frequency of users and servers), providing an effective solution for the non-convex FP problem.

\noindent $\bullet$ The simulation results substantiate the effectiveness of the proposed DASHF algorithm in achieving the joint optimization of user adapter parameter offloading, resource allocation, service-cost ratio, and resource allocation, demonstrating its practical applicability and benefits.

The rest of this paper is organized as follows. The system model and optimization problem formulation are presented in Section \ref{section.System Model}. We propose a novel DASHF algorithm to solve the optimization problem in Section \ref{section.proposed AO technique}. The complexity of the proposed DASHF algorithm is analyzed in Section \ref{section.Complexity Analysis}. The numerical results are provided in Section \ref{section.Numerical Results}. We conclude this paper in Section \ref{section.Conclusion}.

\section{LLM-Empowered MEC System and Optimization Problem Formulation}\label{section.System Model}
In this section, we first introduce the system scenario, then analyze the delay and energy consumption in the system model, then introduce the concept of user service-cost ratio (SCR), and formulate the optimization problem.
\begin{figure}[t]
\centering
\includegraphics[width=0.45\textwidth]{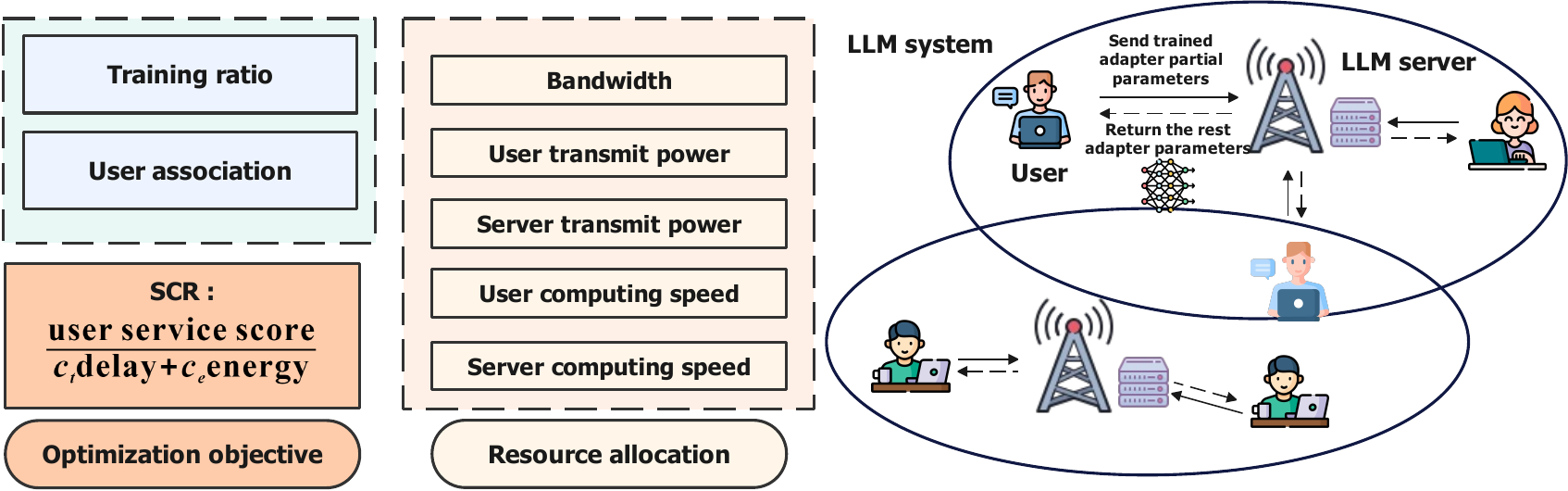}
\caption{\small{Optimizing the SCR of an LLM system with $N$ users and $M$ servers through joint optimization of user association, offloading, and resource allocation.}}\vspace{-15pt}
\label{Fig.system}
\end{figure}
\subsection{System scenario}
As presented in Fig. \ref{Fig.system}, the LLM-based mobile edge computing system contains multiple servers that distribute tailored LLM models or emulators to various mobile users. Given the computational constraints of some mobile users, a hybrid approach is adopted: users train the first few layers locally with their datasets while freezing the other layers. After training, users send these layer parameters of their adapters to the server. Once the server receives those layer parameters, it completes the training of the remaining layers while freezing the layers trained by the users. After training, the server sends the refined parameters back to the users. This collaborative mechanism ensures efficient and personalized LLM services, compensating for individual users' computational limitations.
\subsection{System model}
We consider a system comprising $N$ mobile users and $M$ LLM servers. We use $n$ and $m$ as indices for the $n$-th VR user and the $m$-th LLM server, respectively, where $n \in \mathcal{N} := \{1,2,\cdots,N\}$ and $m \in \mathcal{M} := \{1,2,\cdots,M\}$. Each user is connected to one and only one server; i.e., $\sum_{m\in \mathcal{M}} x_{n,m} = 1$. We introduce indicator variables $x_{n,m}\in\{0,1\}$ to characterize the connection between users and servers; specifically, $x_{n,m}=1$ (resp, $0$) means that the \mbox{$n$-th} user is connected (resp., not connected) to the \mbox{$m$-th} server. For example, if $x_{n,m}=1$, it means that the \mbox{$n$-th} user only connects to the \mbox{$m$-th} server and $x_{n,m'}=0$ for $m{'}\in\mathcal{M} \setminus \{m\}$.
\subsubsection{Time consumption}
We consider frequency-division multiple access (FDMA) so that communication among users and servers would not interfere. The transmission rate from user $n$ to the chosen edge server $m$ is $ r_{n,m}(b_{n,m}, p_{n}) = b_{n,m} \log_2(1+\frac{g_{n,m}p_{n}}{\sigma^2b_{n,m}})$, where $\sigma^2$ is the noise power spectral density, $b_{n,m}$ is the allocated bandwidth between user $n$ and server $m$, $p_n$ is the transmit power of user $n$, $g_{n,m}$ is the channel gain and can be further expressed by $g_{n,m} = h_{n,m}l_{n,m}$, where $h_{n,m}$ is the large-scale slow-fading component (path loss and shadowing) and $l_{n,m}$ is the small-scale Rayleigh fading.

In Parameter-Efficient Fine-Tuning (PEFT) strategies for large language models, the concept is to introduce ``adapters'' - smaller neural network components. These are placed within the model, enabling task-specific customization while largely keeping the pre-trained parameters unchanged. By doing this, there's no need for the exhaustive retraining of the complete model. If we consider inserting an adapter between two layers with dimensions $d_{in}$ and $d_{out}$, the design usually involves: 1. A down-projection from $d_{in}$ to a reduced dimension $d_{adapt}$. 2. An up-projection from $d_{adapt}$ back to $d_{out}$. Accounting for weights and biases in these transformations, the total size of the adapter's parameters, $d$, can be captured as:
$d = d_{in}\times d_{adapt} + d_{adapt}\times d_{out} + d_{adapt} + d_{out}$. For any given user, represented by $n$, the parameter size of the adapter, $d_n$, can differ. This could be due to user-specific requirements or constraints. Hence, $d_n$ links the inherent complexity of the LLM, the architecture of the adapter, and the unique data requirements of user $n$.

\textbf{User training and sending adapter parameter phases.} Based on the above discussion, assume the total adapter parameter size at user $n$ is $d_n$. The adapter parameter size trained by user $n$ is $\varphi_n d_n$, $\varphi_n \in [0,1]$. User $n$ trains the first $\varphi_n d_n$ layer parameters with the local datasets. The training time consumed is $T^{(t_1)}_{n,m} = \frac{t_n \varphi_n d_n e_n}{g_n F_n}$. $t_n$ is the floating point operations (FLOPs) of all tokens for each adapter parameter of the user $n$. $e_n$ is the number of local training epochs. $g_n$ is the available GPU number of user $n$. $F_n$ is the available GPU computation speed of user $n$, whose unit is FLOP. After local training and the user-server connection algorithm (this can be complicated by choosing the nearest neighbor server sets, then choosing the server with the lowest transmission time, and finally finishing all user-server connections), user $n$ transmits $\varphi_n d_n$ adapter parameters to the server $m$. The transmission time from the user $n$ to the server $m$ is $T^{(t_2)}_{n,m} = \frac{x_{n,m}\varphi_n d_n \omega_b}{r_{n,m}}$, where $\omega_b$ is the bits number used to represent each parameter. For example, if we use the ``float32'' floating-point number, $\omega_b$ will be 32.

\textbf{Server training and returning adapter parameter phases.} After receiving the partial adapter parameters $\varphi_n d_n$ from the user $n$, the server $m$ trains the remaining part of the adapter with the shared user datasets and the training delay is $T^{(t_3)}_{n,m} = \frac{t_n (1-\varphi_n) x_{n,m} d_n e_m}{g_m F_m}$. $e_m$ is the number of server training epochs. $g_m$ is the available GPU number of server $m$. $F_m$ is the available GPU computation speed of server $m$. Then, Server $m$ transmits the results to the user $n$, and the delay is $T^{(t_4)}_{n,m} = \frac{x_{n,m}(1-\varphi_n) d_n \omega_b}{r_{m,n}}$, where $r_{m,n}=b_{n,m} \log_2(1+\frac{g_{n,m}p_{m}}{\sigma^2b_{n,m}})$. We assume the path loss and bandwidth between the downlink and uplink are the same. The time consumed on the server side is $T_{s,n,m} = T^{(t_2)}_{n,m} + T^{(t_3)}_{n,m}$. The time consumed on the user side is $T_{u,n,m} =  T^{(t_1)}_{n,m} + T^{(t_4)}_{n}$. Therefore, the total delay will be $T_{total} = \text{max}\{T_{s,n,m} + T_{u,n,m}\}$.
\subsubsection{Energy consumption}
Based on the delay discussion, we then compute the energy consumption in this system. Energy used for user $n$ training the adapter locally can be calculated by $E^{(t_1)}_{n,m} = e_n\kappa_n \varphi_n d_n t_n F^{2}_n$ \cite{zeng2021energy}. $\kappa_n$ is the computational efficiency of user $n$'s GPUs, denoting the power growth rate corresponding to rising computing speeds. Energy used for transmitting data from the \mbox{$n$-th} user to the server $m$ is given as $ E^{(t_2)}_{n,m} = p_n T^{(t_2)}_{n,m} = p_n \frac{x_{n,m}\varphi_n d_n\omega_b}{r_{n,m}}$.  Energy for server training $\varphi_n d_n$ adapter parameters is given as $E^{(t_3)}_{n,m} = e_m\kappa_m x_{n,m} (1-\varphi_n) d_n t_n F_{m}^{2}$. $\kappa_m$ is the computational efficiency of server $m$'s GPUs. Energy caused by \mbox{$m$-th} server transmitting trained adapter parameters to user $n$ is $E^{(t_4)}_{n,m} = p_m T^{(t_4)}_{n,m} = p_m \frac{x_{n,m} (1-\varphi_n) d_n\omega_b}{r_{m,n}}$. Thus, the total energy consumption can be formulated as $E_{total}  = \sum_{n \in \mathcal{N},m\in \mathcal{M}} (E_{u,n,m} + E_{s,n,m})$, where $E_{u,n,m} = E^{(t_1)}_{n,m} + E^{(t_2)}_{n,m}$ and $E_{s,n,m} = E^{(t_3)}_{n,m} + E^{(t_4)}_{n,m}$.
\subsubsection{User service score}
We denote the service score of user $n$ that is connected to server $m$ as:
\begin{talign}
    v_{n,m} = \varpi_1 \ln[1+\varpi_2 (\frac{p_{m}}{p_{max}^{(m)}} + \frac{F_{n,m}}{F_{max}^{(m)}} + \frac{b_{n,m}}{b_{max}})],\label{trustvalue}
\end{talign}
where $\varpi_1$ determines the range of function value, $\varpi_2$ is used for normalization of $(\frac{p_m}{p_{max}^{(m)}} + \frac{F_{n,m}}{F_{max}^{(m)}} + \frac{b_{n,m}}{b_{max}})$. This function is jointly concave of $p_m$, $F_{n,m}$, and $b_{n,m}$ \cite{yang2015incentive}. This user service score function is effective and sensitive in all value ranges of $(\frac{p_m}{p_{max}^{(m)}} + \frac{F_{n,m}}{F_{max}^{(m)}} + \frac{b_{n,m}}{b_{max}})$, which can describe each user's resources obtained from the server. 
\subsection{Optimization Problem}
We define user connection $\bm{x} := (x_{n,m})|_{n \in \mathcal{N}, m \in \mathcal{M} }$, offloading ratio $\bm{\varphi} := (\varphi_n)|_{n \in \mathcal{N}}$, bandwidth $\bm{b} := (b_{n,m})|_{n \in \mathcal{N}, m \in \mathcal{M} }$, transmission power $\bm{p_u} := (p_n)|_{n \in \mathcal{N}}$, $\bm{p_s} := (p_m)|_{m \in \mathcal{M} }$, GPU computation speed  $\bm{f_u} := (F_n)|_{n \in \mathcal{N}}$ ,and $\bm{f_s} := (F_{m})|_{m \in \mathcal{M} }$. Our goal is to maximize the user service-cost ratio (SCR):
\begin{talign}
    &\frac{\mathcal{V}}{\omega_t T_{total} + \omega_e E_{total}}
     = \nonumber \\ &\frac{\sum_{n \in \mathcal{N}, m \in \mathcal{M}} x_{n,m} v_{n,m}}{\omega_t (T_{s,n,m} + T_{u,n,m}) + \omega_e (\hspace{-14pt}\sum\limits_{n \in \mathcal{N},m\in \mathcal{M}}\hspace{-14pt} (E^{(t_1)}_{n,m}+E^{(t_2)}_{n,m}+E^{(t_3)}_{n,m}+E^{(t_4)}_{n,m}))},
\end{talign}
where $\omega_t$ and $\omega_e$ represent the weight values of delay and energy, respectively. In order to linearize the ``maximize" term of $T_{total}$, we add an auxiliary variable $T$,  which is constrained to be greater than or equal to $T_{total}$. Besides, we utilize Dinkelbach's Algorithm~\cite{dinkelbach1967nonlinear} by adding an additional variable $y$, which is obtained from the SCR value in the previous iteration. Then, the fractional programming in the trust-cost ratio is transformed into the following problem:
\begin{subequations}\label{prob1}
\begin{align}
&\max\limits_{\bm{x},\bm{\varphi},\bm{b},\bm{p_u},\bm{p_s},\bm{f_u},\bm{f_s},T}  \nonumber \\& \hspace{-20pt} \big\{ \!\!\!\!\sum\limits_{n \in \mathcal{N}, m \in \mathcal{M}} \!\!\!\![x_{n,m} v_{n,m}  \hspace{-2pt}-\hspace{-2pt}y \omega_e (E_{u,n,m} \!\!+\!\! E_{s,n,m})]\big\}\hspace{-2pt}-\hspace{-2pt} y \omega_t T\tag{\ref{prob1}}\\
\text{s.t.} \quad & x_{n,m} \in \{0,1\}, \forall n,m \label{x_constr1}\\       & \sum\limits_{m\in \mathcal{M}} x_{n,m} = 1, \forall n \label{x_constr2}\\
         & \varphi_n \in [0,1], \forall n \label{varphi_constr}\\
         & \sum\limits_{n\in \mathcal{N}} x_{n,m} b_{n,m} \leq b_{max},\forall m \label{x_b_constr}\\
         & p_n \leq p^{(n)}_{max}, \forall n \label{pu_constr}\\
         & \sum\limits_{n\in \mathcal{N}} x_{n,m} p_{n,m} \leq p^{(m)}_{max}, \forall m \label{x_ps_constr}\\
         & F_n \leq F^{(n)}_{max}, \forall n \label{fu_constr}\\
         & \sum\limits_{n\in \mathcal{N}} x_{n,m} F_{m} \leq F^{(m)}_{max}, \forall m \label{x_fs_constr}\\
         & b_{n,m} \!\geq\! 0, p_n \!\geq \!0, p_{n,m} \!\geq\! 0, F_n \!\geq\! 0, F_{m} \!\geq \!0, \forall n,m \label{nonnegative_constr}\\
         & T_{s,n,m} + T_{u,n,m} \leq T, \forall n,m\label{T_constr}
\end{align}
\end{subequations}
Based on Dinkelbach's Algorithm, we iteratively optimize $y$ and problem (\ref{prob1}). Specifically, at the $i$-th iteration, given $y^{(i-1)}$, we first obtain $\bm{x}^{(i)},\bm{\varphi}^{(i)},\bm{b}^{(i)},\bm{p_u}^{(i)},\bm{p_s}^{(i)},\bm{f_u}^{(i)},\bm{f_s}^{(i)}, T^{(i)}$ by solving the optimization problem \ref{prob1}; then we calculate $y^{(i)}$ with the given $\bm{x}^{(i)},\bm{\varphi}^{(i)},\bm{b}^{(i)},\bm{p_u}^{(i)},\bm{p_s}^{(i)},\bm{f_u}^{(i)},\bm{f_s}^{(i)}, T^{(i)}$. Repeat the above operations until the solutions converge. In the following section, we consider using the alternating optimization method (AO) to tackle the complex problem (\ref{prob1}). 

\textbf{Roadmap of the whole algorithm.} First, we decompose the outer fractional structure of the original SCR problem using Dinkelbach algorithm and sequentially optimize $\bm{x},\bm{\varphi}$ and $\bm{b},\bm{p_u},\bm{p_s},\bm{f_u},\bm{f_s}$ using the AO method. In the first step of AO, we fix $\bm{b},\bm{p_u},\bm{p_s},\bm{f_u},\bm{f_s}$ and optimize $\bm{x},\bm{\varphi},T$. We transform the optimization problem in the first step of AO into a Quadratically Constrained Quadratic Program (QCQP) and solve it using Semidefinite Relaxation (SDR) and the Hungarian algorithm. In the second step of AO, we fix $\bm{x},\bm{\varphi}$ and optimize $\bm{b},\bm{p_u},\bm{p_s},\bm{f_u},\bm{f_s}, T$. During the optimization in the second step of AO, we propose a new fractional programming method to transform this non-convex problem into a convex one. Finally, we calculate $y$ based on the obtained solutions and repeat the aforementioned process until $y$ converges. In this algorithm, since we utilize \underline{\textbf{D}}inkelbach's algorithm, \underline{\textbf{a}}lternating optimization, \underline{\textbf{s}}emidefinite relaxation, \underline{\textbf{H}}ungarian algorithm, and \underline{\textbf{f}}ractional programming, we refer to this algorithm as the \textbf{DASHF Algorithm}.

\section{Our proposed DASHF Algorithm to solve the optimization problem}\label{section.proposed AO technique}
Assuming that $y$ is given, we need to optimize $\bm{x},\bm{\varphi},\bm{b},\bm{p_u},\bm{p_s},\bm{f_u},\bm{f_s}, T$. In the outermost loops, we iteratively optimize $y$; In the innermost loops, we iteratively optimize $\bm{x},\bm{\varphi},\bm{b},\bm{p_u},\bm{p_s},\bm{f_u},\bm{f_s}, T$. However, it is still difficult to optimize them in parallel. Thus, we consider operating two inner AO steps to solve it. At the $i$-th iteration, 
\begin{enumerate}
    \item Optimize $\vect{x}$, $\vect{\varphi}$, $T$, given $\vect{b}$, $\vect{p_u}$, $\vect{p_s}$, $\vect{f_u}$, $\vect{f_s}$. Assuming that $\vect{b}^{(i-1)}$, $\vect{p_u}^{(i-1)}, $$\vect{p_s}^{(i-1)}$, $\vect{f_u}^{(i-1)}$, $\vect{f_s}^{(i-1)}$, $y^{(i-1)}$ are given, we optimize $\vect{x}^{(i)}$, $\vect{\varphi}^{(i)}$, $T^{(i)}$.
    \item Optimize $\vect{b}$,$\vect{p_u}$,$\vect{p_s}$,$\vect{f_u}$,$\vect{f_s}$,$T$, given $\vect{x}$, $\vect{\varphi}$. Assuming that $\vect{x}^{(i-1)}$, $\vect{\varphi}^{(i-1)}$,$y^{(i-1)}$ are given, we optimize $\vect{b}^{(i)}$,$\vect{p_u}^{(i)}$,$\vect{p_s}^{(i)}$,$\vect{f_u}^{(i)}$,$\vect{f_s}^{(i)}$,$T^{(i)}$.
\end{enumerate}

\subsection{AO Part 1: Optimize \texorpdfstring{$\vect{x}, \vect{\varphi},T$}{}, given \texorpdfstring{$\vect{b},\vect{p_u},\vect{p_s},\vect{f_u},\vect{f_s}$}{}}\label{section-AO-part1}
Given $\bm{b},\bm{p_u},\bm{p_s},\bm{f_u},\bm{f_s},T$, we optimize $\bm{x},\bm{\varphi},T$. The optimization problem will be:
\begin{subequations}\label{prob2}
\begin{align}
\!\!\!\max\limits_{\bm{x},\bm{\varphi},T}\!\!\!\! & \sum\limits_{n \in \mathcal{N}, m \in \mathcal{M}} \!\!\!\!\!\![x_{n,m} v_{n,m} \!-\! y\omega_e (E_{u,n,m} \!+ \!E_{s,n,m})]-y \omega_t T \tag{\ref{prob2}}\\
\text{s.t.}~ & (\text{\ref{x_constr1}}), (\text{\ref{x_constr2}}), (\text{\ref{varphi_constr}}),(\text{\ref{x_b_constr}}),(\text{\ref{x_ps_constr}}),(\text{\ref{x_fs_constr}}),(\text{\ref{T_constr}}).\nonumber
\end{align}
\end{subequations}
$x_{n,m}$ are binary variables and this is a mixed-integer nonlinear programming problem. We rewrite $x_{n,m}\in \{0,1\}, \forall n,m$ as $x_{n,m}(x_{n,m}-1)=0, \forall n,m$. The optimization problem will be rewritten as:
\begin{subequations}\label{prob3}
\begin{align}
&\!\!\!\!\max\limits_{\bm{x},\bm{\varphi},T}\!\!\!\!\!  \sum\limits_{n \in \mathcal{N}, m \in \mathcal{M}} \!\!\!\!\!\!\!\![x_{n,m} v_{n,m} \!-\! y\omega_e (E_{u,n,m} \!+ \!E_{s,n,m})]-y \omega_t T \tag{\ref{prob3}}\\
&\text{s.t.} \quad  x_{n,m}(x_{n,m}-1)=0, \forall n,m \label{x_constr1_new}\\
           & \quad\quad(\text{\ref{x_constr2}}), (\text{\ref{varphi_constr}}),(\text{\ref{x_b_constr}}),(\text{\ref{x_ps_constr}}),(\text{\ref{x_fs_constr}}),(\text{\ref{T_constr}}).\nonumber
\end{align}
\end{subequations}
We substitute the expression of $E_{u,n,m}+E_{s,n,m}$ into problem (\ref{prob3}) and convert the $\max$ problem in problem (\ref{prob3}) to a $\min$ problem. Besides, let $G_{n,m} \!= \!y \omega_e (\frac{p_n d_n\omega_b}{r_{n,m}} - \frac{p_m d_n\omega_b}{r_{m,n}} - e_m \kappa_m d_n t_n F_m^2)$, $A_{n,m} = y \omega_e(e_n \kappa_n d_n t_n F_n^2)$, and $B_{n,m} = y \omega_e(\frac{p_m d_n\omega_b}{r_{m,n}} + e_m \kappa_m d_n t_n F_m^2)-v_{n,m}$.  Let $\mathbf{A} = [A_{n}]|_{n \in \mathcal{N}}$, $\mathbf{B} = [B_{n,m}]|_{n \in \mathcal{N},m\in \mathcal{M}}$, and $\mathbf{G} = [G_{n,m}]|_{n \in \mathcal{N},m\in \mathcal{M}}$. The optimization problem (\ref{prob3}) can be rewritten as
\begin{subequations}\label{prob5}
\begin{align}
\mathcal{P}_{1}: &\min\limits_{\bm{x},\bm{\varphi},T}  y\omega_t T \!\!+\!\! \sum\limits_{n \in \mathcal{N}} A_{n}\varphi_n +\!\!\!\!\!\sum\limits_{n \in \mathcal{N}, m \in \mathcal{M}}\!\!\!\!\! (B_{n,m} x_{n,m}   + \nonumber \\ &\quad\quad\quad\quad G_{n,m}x_{n,m}\varphi_n) \tag{\ref{prob5}}\\
&\text{s.t.} \quad  (\text{\ref{x_constr1_new}}),(\text{\ref{x_constr2}}), (\text{\ref{varphi_constr}}),(\text{\ref{x_b_constr}}),(\text{\ref{x_ps_constr}}),(\text{\ref{x_fs_constr}}),(\text{\ref{T_constr}}).\nonumber
\end{align}
\end{subequations}
This is a quadratically constrained quadratic program (QCQP) problem. Then, we need to get the standard form of the QCQP problem. Let $\boldsymbol{B}:=(b_{1,1},\cdots,b_{n,m})^\intercal$, $\boldsymbol{P}:=(p_{1,1},\cdots,p_{n,m})^\intercal$, and $\boldsymbol{F}:=(F_{1,1},\cdots,F_{n,m})^\intercal$.
\begin{theorem}
There're matrices $\mathbf{Q}$, $\mathbf{P}_0$, $\mathbf{W}_0$, $\mathbf{P^{(\text{T1})}}$, $\mathbf{P^{(\text{T2})}}$, and $\mathbf{P^{(\text{T3})}}$ that convert Problem $\mathcal{P}_1$ into QCQP Problem $\mathcal{P}_2$:
\begin{subequations}\label{prob6}
\begin{align}
\mathcal{P}_{2}: &\min\limits_{\bm{Q},T}\quad  \boldsymbol{Q}^\intercal\mathbf{P}_0 \boldsymbol{Q}+\mathbf{W}_{0}^\intercal\boldsymbol{Q}+\mathbf{W}_{1}^\intercal\boldsymbol{Q}+y\omega_t T\tag{\ref{prob6}}\\
\hspace{-15pt} \text{s.t.} \quad & \hspace{-7pt} \text{diag}(\boldsymbol{e}_{N+1,NM+N}^\intercal\boldsymbol{Q})(\text{diag}(\boldsymbol{e}_{N+1,NM+N}^\intercal\boldsymbol{Q})-\mathbf{I})=\mathbf{0} \label{x_constr1_qcqp}\\       & \text{diag}(\boldsymbol{e}_{\overline{1},\overline{M}}^\intercal\boldsymbol{e}_{N+1,NM+N}^\intercal\boldsymbol{Q})=\mathbf{I} \label{x_constr2_qcqp}\\
         & \text{diag}(\boldsymbol{e}_{1,N}^\intercal\boldsymbol{Q})\leq\mathbf{I} \label{varphi_constr1_qcqp}\\
         & \text{diag}(\boldsymbol{e}_{1,N}^\intercal\boldsymbol{Q})\geq\mathbf{0} \label{varphi_constr2_qcqp}\\
         & \boldsymbol{B}^\intercal\boldsymbol{e}_{N+1,NM+N}^\intercal\boldsymbol{Q}-B_{max} \leq 0 \label{x_b_constr_qcqp}\\
         & \boldsymbol{P}^\intercal\boldsymbol{e}_{N+1,NM+N}^\intercal\boldsymbol{Q}-P_{max}^{(m)} \leq 0 \label{x_ps_constr_qcqp}\\
         & \boldsymbol{F}^\intercal\boldsymbol{e}_{N+1,NM+N}^\intercal\boldsymbol{Q}-F_{max}^{(m)} \leq 0 \label{x_fs_constr_qcqp}\\
         & \boldsymbol{Q}^\intercal\mathbf{P}^{(\text{T1})} \boldsymbol{Q} + {\mathbf{P}^{(\text{T2})}}^\intercal \boldsymbol{Q} + {\mathbf{P}^{(\text{T3})}}^\intercal \boldsymbol{Q} \leq T, \label{T_constr_qcqp}
\end{align}
\end{subequations}
\end{theorem}
\begin{proof}
See Section \uppercase\expandafter{\romannumeral3} Part A in the online version \cite{qian2023user}.
\end{proof}
Problem (\ref{prob6}) is the standard QCQP form. However, it is still non-convex. Then, we need to utilize the semidefinite programming (SDP) method to transform this QCQP problem into a semidefinite relaxation (SDR) problem. 
\begin{theorem}
There're matrices $\mathbf{S}$, $\mathbf{P}_1$, $\mathbf{P}_2$, $\mathbf{P}_3$, $\mathbf{P}_4$, $\mathbf{P}_5$, $\mathbf{P}_6$, $\mathbf{P}_7$, and $\mathbf{P}_8$ that transform Problem $\mathcal{P}_2$ into SDR Problem $\mathcal{P}_3$:
\begin{subequations}\label{prob7}
\begin{align}
\mathcal{P}_{3}: \min\limits_{\bm{S},T}\quad & \text{Tr}(\mathbf{P}_1 \mathbf{S})\tag{\ref{prob7}}\\
\text{s.t.} \quad & \text{Tr}(\mathbf{P}_2 \mathbf{S})=0 \label{x_constr1_sdr}\\       & \text{Tr}(\mathbf{P}_3 \mathbf{S})=0 \label{x_constr2_sdr}\\
         & \text{Tr}(\mathbf{P}_4 \mathbf{S})\leq0 \label{varphi_constr_sdr}\\
         & \text{Tr}(\mathbf{P}_5 \mathbf{S})\leq0 \label{x_b_constr_sdr}\\
         & \text{Tr}(\mathbf{P}_6 \mathbf{S})\leq0 \label{x_ps_constr_sdr}\\
         & \text{Tr}(\mathbf{P}_7 \mathbf{S})\leq0 \label{x_fs_constr_sdr}\\
         & \text{Tr}(\mathbf{P}_8 \mathbf{S})\leq T \label{T_constr_sdr}\\
         & \mathbf{S}\succeq0, \label{S_constr_sdr}
\end{align}
\end{subequations}
\end{theorem}
\begin{proof}
    See Section \uppercase\expandafter{\romannumeral3} Part A in the online version \cite{qian2023user}.
\end{proof}
The constraints $(\text{\ref{x_constr1_qcqp}})$, $(\text{\ref{x_constr2_qcqp}})$, $(\text{\ref{varphi_constr1_qcqp}})$, $(\text{\ref{varphi_constr2_qcqp}})$, $(\text{\ref{x_b_constr_qcqp}})$, $(\text{\ref{x_ps_constr_qcqp}})$, $(\text{\ref{x_fs_constr_qcqp}})$, $(\text{\ref{T_constr_qcqp}})$ in $\mathcal{P}_{2}$ are transformed into the constraints $(\text{\ref{x_constr1_sdr}}),(\text{\ref{x_constr2_sdr}}), (\text{\ref{varphi_constr_sdr}}),(\text{\ref{S_constr_sdr}}), (\text{\ref{x_b_constr_sdr}}),(\text{\ref{x_ps_constr_sdr}}),(\text{\ref{x_fs_constr_sdr}}),(\text{\ref{T_constr_sdr}})$ in $\mathcal{P}_{3}$, respectively. Drop the constraint $\text{rank}(\mathbf{S})=1$ and the objective function and the constraints are all convex. Then this SDR problem will be solved in polynomial time by common convex solvers. By solving this SDR problem, we can get a continuous solution of $\vect{Q}$. However, this solution is the lower bound of the optimal solution and it may not guarantee the constraint $\text{rank}(\mathbf{S})=1$. Therefore, we need to use rounding techniques to recover the solution. The latter $NM$ elements in $\vect{Q}$ is $x_{n,m}$, for all $n \in \mathcal{N}, m \in \mathcal{M}$, which means that user $n$ is fractional connected to server $m$. Then, find all user $n$ that $\sum_{m \in \mathcal{M}} x_{n,m} > 1$. For these users, modify $x_{n,m}$ as $\frac{x_{n,m}}{|\sum_{m \in \mathcal{M}} x_{n,m}|}$. Use the Hungarian algorithm \cite{dai2018joint} with augmented zero vectors to find the best matching with the maximum weight and denote this matching as a set $\mathcal{X}_{matching}$. For nodes $n$ and $m$ in $\mathcal{X}_{matching}$, set $x_{n,m} = 1$, otherwise $x_{n,m} = 0$. Representing this as $\vect{x}_{\sharp}$, substitute it into Problem (\ref{prob6}) to derive optimal $\vect{\varphi}$. 

\subsection{AO Part 2: Optimize \texorpdfstring{$\vect{b},\vect{p_u},\vect{p_s},\vect{f_u},\vect{f_s},T$}{}, given \texorpdfstring{$\vect{x}, \vect{\varphi}$}{}}\label{section-AO-part2}
We first define $\chi_1(p_n):=p_n x_{n,m}\varphi_n d_n\omega_b$, $\chi_2(p_m):=p_m x_{n,m} (1-\varphi_n) d_n\omega_b$. Let
\begin{talign}
    &\mathcal{F}(\bm{b},\bm{p_s},\bm{f_u},\bm{f_s},T):=\!\!\sum_{n \in \mathcal{N}, m \in \mathcal{M}}x_{n,m} v_{n,m} -y \omega_t T \nonumber \\& -y\omega_e \big\{\sum_{n \in \mathcal{N}} e_n \kappa_n \varphi_n d_n t_n F^{2}_n \nonumber\\& + \sum_{n\in \mathcal{N}, m\in \mathcal{M}} \!e_m \kappa_m x_{n,m} (1-\varphi_n) d_n t_n F_{n,m}^2\big\},
\end{talign}
and $\mathcal{F}$ is in short for $\mathcal{F}(\bm{b},\bm{p_s},\bm{f_u},\bm{f_s},T)$. It's easy to justify that $\mathcal{F}$ is concave. Given $\vect{x}$ and $\vect{\varphi}$, the remaining optimization problem is shown as (\ref{prob8}):
\begin{align}
&\max\limits_{\bm{b},\bm{p_u},\bm{p_s},\bm{f_u},\bm{f_s},T}\quad \mathcal{F}-y \omega_e \sum_{n \in \mathcal{N}, m \in \mathcal{M}}(\frac{\chi_1(p_n)}{r_{n,m}}+\! \frac{\chi_2(p_m)}{r_{m,n}})\label{prob8}\\
&\text{s.t.} \quad 
         (\text{\ref{x_b_constr}}), (\text{\ref{pu_constr}}), (\text{\ref{x_ps_constr}}), (\text{\ref{fu_constr}}), (\text{\ref{x_fs_constr}}), (\text{\ref{nonnegative_constr}}), (\text{\ref{T_constr}}),\nonumber
\end{align}
where $\frac{\chi_1(p_n)}{r_{n,m}} + \frac{\chi_2(p_m)}{r_{m,n}}$ is non-convex or concave. Then, according to the fractional programming technique introduced in Section \uppercase\expandafter{\romannumeral4} in \cite{zhao2023human}, we let $z_{1,n,m}=\frac{1}{2 \chi_1(p_n) r_{n,m}}$ and $z_{2,n,m}=\frac{1}{2 \chi_2(p_m) r_{m,n}}$. The optimization objective function can be expressed as
\begin{talign}\label{fp_trans_prob}
&\mathcal{F}-y \omega_e \sum_{n \in \mathcal{N}, m \in \mathcal{M}}[\chi_1(p_n)^2z_{1,n,m}+\frac{1}{4 r_{n,m}^2 z_{1,n,m}} \!\!+ \!\!\chi_2(p_m)^2 \nonumber\\ & \cdot z_{2,n,m} + \frac{1}{4 r_{m,n}^2 z_{2,n,m}}].
\end{talign}
The transformation optimization problem is shown as (\ref{prob9}):
\begin{align}
&\max\limits_{\bm{b},\bm{p_u},\bm{p_s},\bm{f_u},\bm{f_s},T}\mathcal{F}-y \omega_e \sum_{n \in \mathcal{N}, m \in \mathcal{M}}[\chi_1(p_n)^2z_{1,n,m}\nonumber \\&+\frac{1}{4 r_{n,m}^2 z_{1,n,m}} +\chi_2(p_m)^2 z_{2,n,m}  + \frac{1}{4 r_{m,n}^2 z_{2,n,m}}]\label{prob9}\\
& \quad \text{s.t.}
         (\text{\ref{x_b_constr}}), (\text{\ref{pu_constr}}), (\text{\ref{x_ps_constr}}), (\text{\ref{fu_constr}}), (\text{\ref{x_fs_constr}}), (\text{\ref{nonnegative_constr}}), (\text{\ref{T_constr}}).\nonumber
\end{align}
If $\vect{z}_1 = (z_{1,1,1},\cdots,z_{1,n,m})^\intercal$ and $\vect{z}_2 = (z_{2,1,1},\cdots,z_{2,n,m})^\intercal$ is given, the objective function (\ref{prob9}) is concave. At the $i$-th iteration, $[y^{(i)}$, $\vect{z}_1^{(i)}$, $\vect{z}_2^{(i)}]$ are first calculated with the solution $[\vect{b}^{(i-1)},\vect{p_u}^{(i-1)},\vect{p_s}^{(i-1)},\vect{f_u}^{(i-1)},\vect{f_s}^{(i-1)}]$. Then, $[\vect{b}^{(i)},\vect{p_u}^{(i)},\vect{p_s}^{(i)},\vect{f_u}^{(i)},\vect{f_s}^{(i)}]$ can be obtained by solving the concave problem (\ref{prob9}) with $[y^{(i)}, \vect{z}_1^{(i)}, \vect{z}_2^{(i)}]$. Thus, the optimization problem is concave and can be solved by convex tools. The DASHF algorithm process is detailed in \mbox{Algorithm \ref{algo:AO}}.
\begin{algorithm}
\caption{DASHF Algorithm.}
\label{algo:AO}

Initialize $i \leftarrow -1$ and for all $n \in \mathcal{N}, m \in \mathcal{M}$: $b_{n,m}^{(0)} = \frac{b_{\max}}{N}$, $p_n^{(0)}=p_n$, $p_m^{(0)} = \frac{p_{\max}^{(m)}}{N}$, $F_n^{(0)} = F_{max}^{(n)}$, $F_{n,m}^{(0)} = \frac{F_{n,m}^{(m)}}{N}$, $\vect{x}^{(0)} = (\vect{e_1},\cdots,\vect{e_M})^\intercal$, $\varphi_n^{(0)}=0.5$;

$y^{(0)} \leftarrow \sum\limits_{n \in \mathcal{N}, m \in \mathcal{M}} \frac{x_{n,m}^{(0)} v_{n,m}}{\omega_t T^{(0)} + \omega_e (E_{u,n,m}^{(0)} + E_{s,n,m}^{(0)})}$;

\Repeat{$\frac{y^{(i+1)}}{y^{(i)}}- 1 \leq \epsilon$}{

Let $i \leftarrow i+1$;

Obtain $[\vect{x}^{(i+1)},\vect{\varphi}^{(i+1)}]$ by AO part 1;

Obtain $[\boldsymbol{b}^{(i+1)},\boldsymbol{p_u}^{(i+1)},\boldsymbol{p_s}^{(i+1)},\boldsymbol{f_u}^{(i+1)},\boldsymbol{f_s}^{(i+1)}]$ by AO part 2;

Update $y^{(i+1)} \leftarrow \sum\limits_{n \in \mathcal{N}, m \in \mathcal{M}} \frac{x_{n,m}^{(i+1)} v_{n,m}}{\omega_t T^{(i+1)} + \omega_e (E_{u,n,m}^{(i+1)} + E_{s,n,m}^{(i+1)})}$;
}

Set$[\vect{x}^{(i+1)},\vect{\varphi}^{(i+1)},\boldsymbol{b}^{(i+1)},\boldsymbol{p_u}^{(i+1)},\boldsymbol{p_s}^{(i+1)},\boldsymbol{f_u}^{(i+1)},\boldsymbol{f_s}^{(i+1)}]$ as a solution to Problem (\ref{prob1});

Return $[\vect{x}^{(i+1)},\vect{\varphi}^{(i+1)},\boldsymbol{b}^{(i+1)},\boldsymbol{p_u}^{(i+1)},\boldsymbol{p_s}^{(i+1)},\boldsymbol{f_u}^{(i+1)},\boldsymbol{f_s}^{(i+1)}]$ as a solution $[\vect{x}^{\star},\vect{\varphi}^{\star},\boldsymbol{b}^{\star},\boldsymbol{p_u}^{\star},\boldsymbol{p_s}^{\star},\boldsymbol{f_u}^{\star},\boldsymbol{f_s}^{\star}]$ to Problem~(\ref{prob1}).
\end{algorithm}

\begin{figure*}[t]
\vspace{-0.3cm}
\subfigure[Convergence of Algorithm AO-Part 1]{\includegraphics[width=.33\textwidth]{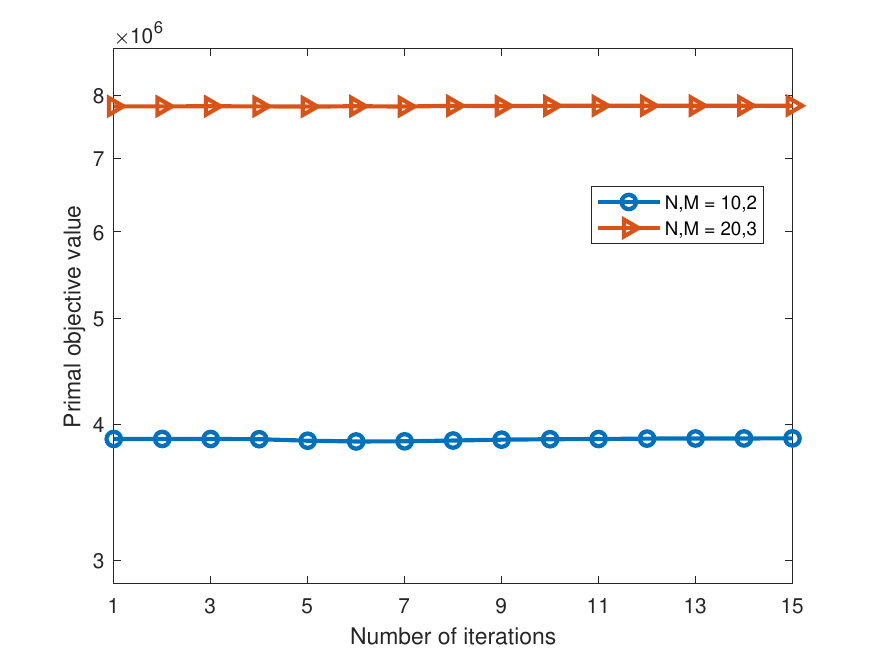}\label{fig:Convergence_of_AO-Part1}}
\subfigure[Convergence of Algorithm AO-Part 2]{\includegraphics[width=.33\textwidth]{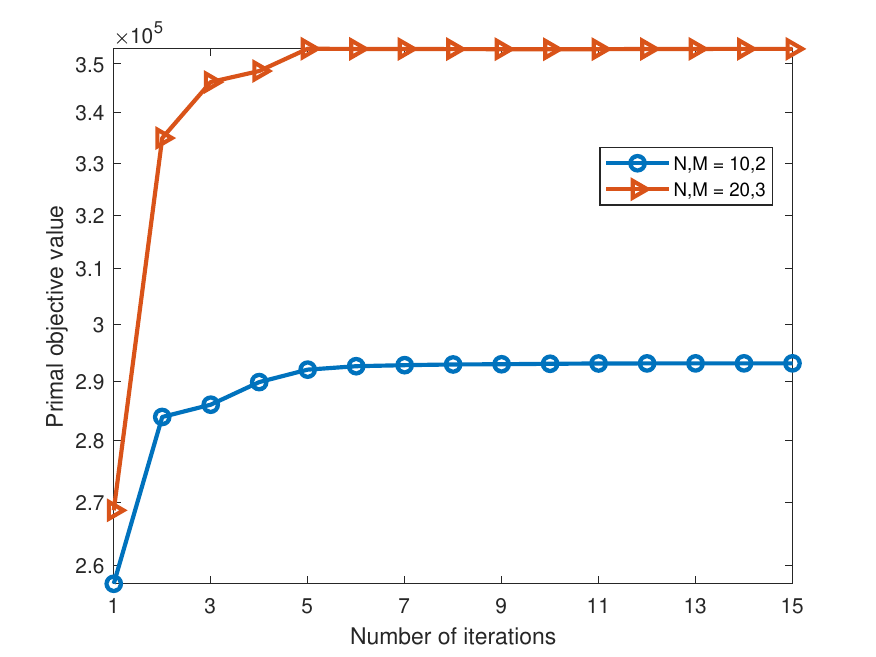}\label{fig:Convergence_of_AO-Part2}}
\subfigure[Convergence of Algorithm DASHF]{\includegraphics[width=.33\textwidth]{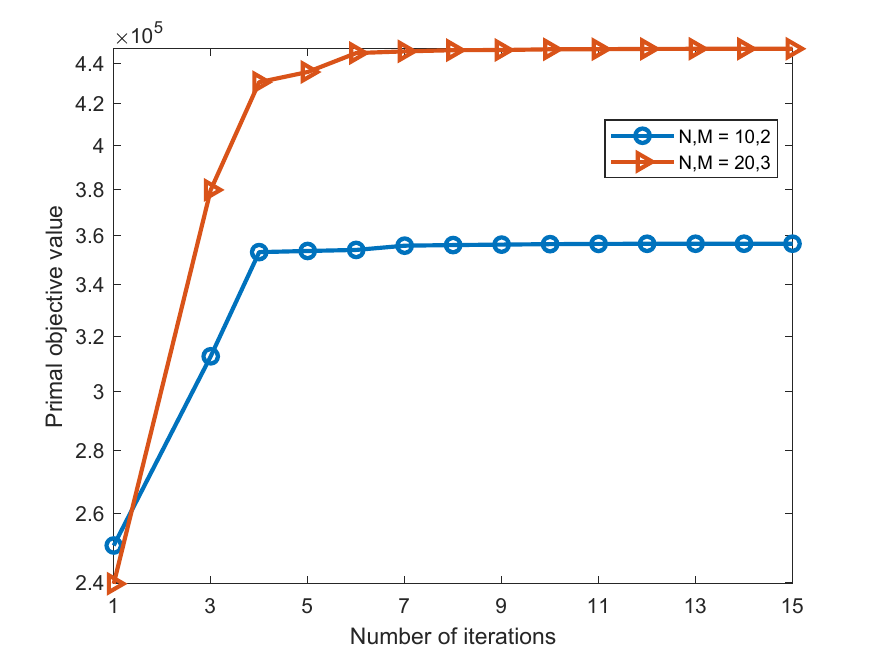}\label{fig:Convergence_of_AO}} \vspace{-10pt}
\caption{Convergence of the proposed Algorithms} 
\end{figure*}

\begin{figure*}[t]
\vspace{-0.3cm}
\subfigure[Resource consumption and SCR of baselines and proposed DASHF Algorithm]{\includegraphics[width=.33\textwidth]{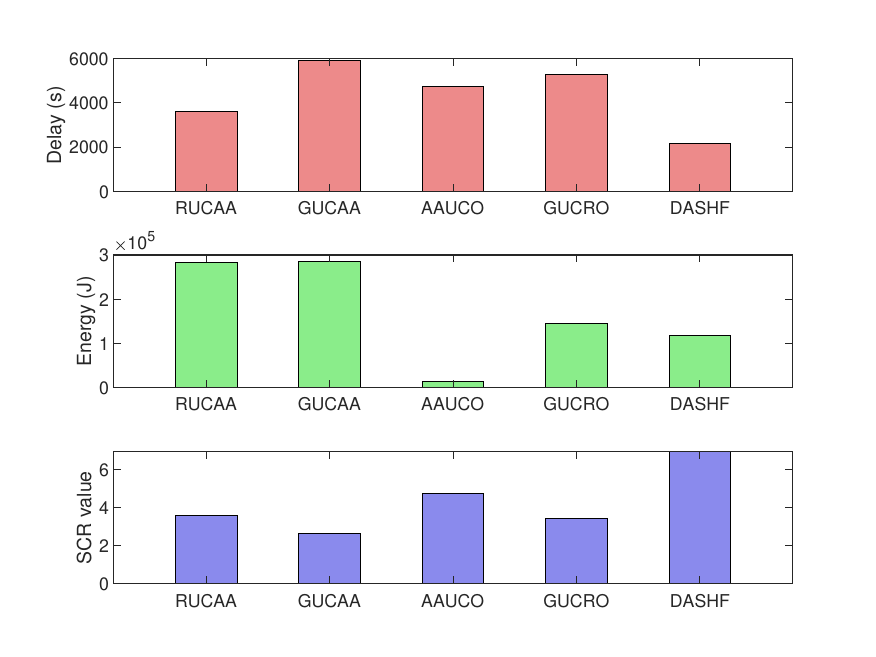}\label{Fig.AO performance}}
\subfigure[SCR under different total bandwidth]{\includegraphics[width=.33\textwidth]{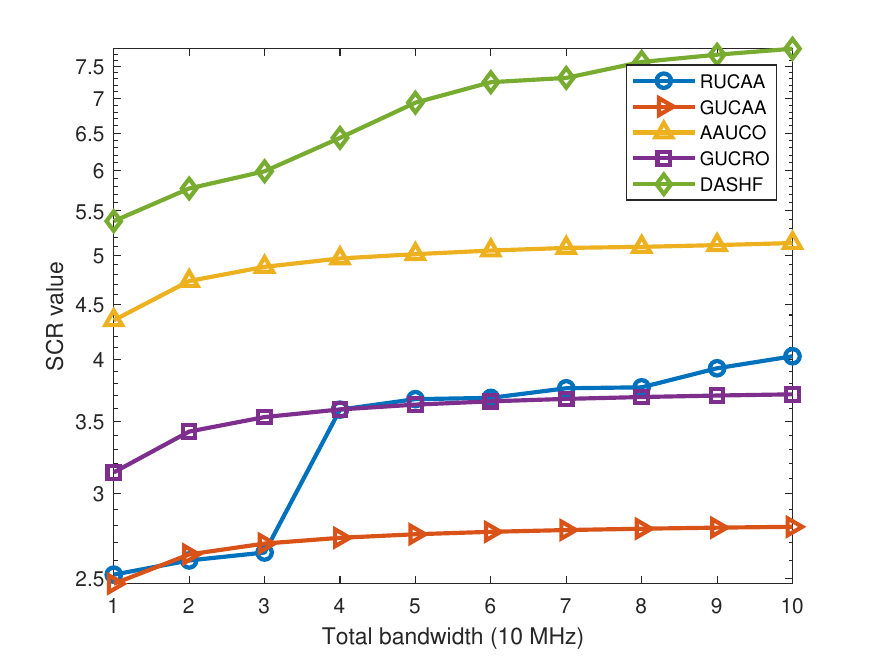}\label{Fig.Comparison_of_b_max}}
\subfigure[Resources consumption and SCR under different $(\omega_{t},\omega_{e})$]{\includegraphics[width=.33\textwidth]{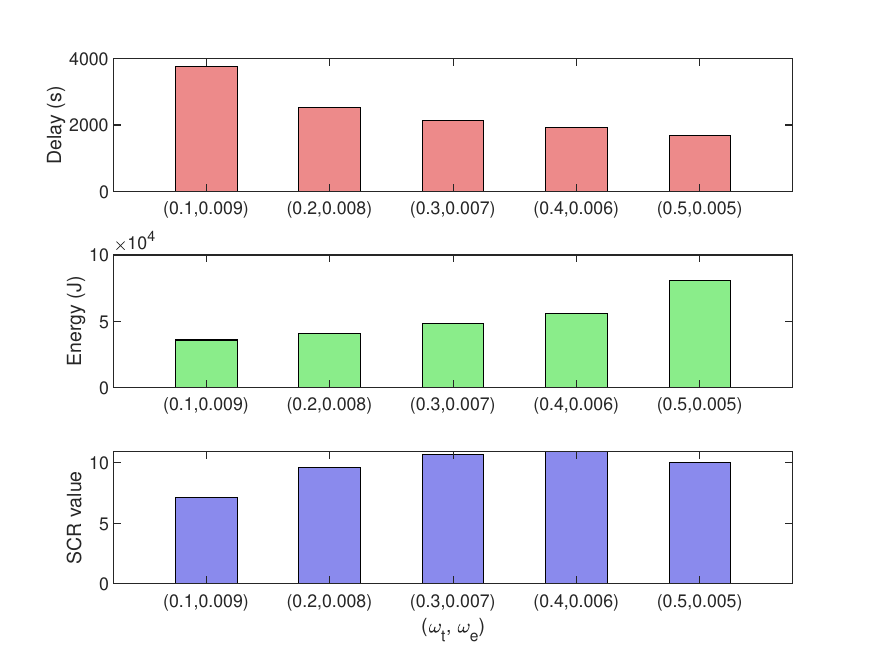}\label{Fig.Comparison_of_omega}}\vspace{-10pt}
\caption{SCR comparisons under different cases} \vspace{-15pt}
\end{figure*}
\section{Complexity Analysis of the DASHF Algorithm}\label{section.Complexity Analysis}
In \ref{section-AO-part1}, at each iteration, there are $a_1 = N(M+1)+1$ variables and $b_1 = N(2M+2)+3M$ constraints. The complexity of the Hungarian algorithm is $\mathcal{O}(N^3)$. The worst-case complexity of \ref{section-AO-part1} is $\mathcal{O}(N^3+(a_1^2b_1+a_1^3)b_1^{0.5}\log(\frac{1}{\epsilon_1}))$ with a given solution accuracy $\epsilon_1 > 0$ \cite{dai2018joint}. In \ref{section-AO-part2}, at each iteration, there are $a_2 = N(3M+2)+1$ variables and $b_2 = N(M+2)+3M$ constraints. The worst-case complexity of \ref{section-AO-part2} is $\mathcal{O}(a_2^2b_2+a_2^3)b_2^{0.5}\log(\frac{1}{\epsilon_2})$ with a given solution accuracy $\epsilon_2 > 0$. To summarize, if the whole algorithm takes $\mathcal{I}$ iterations, the whole complexity is $\mathcal{I}\mathcal{O}(N^3+(a_1^2b_1+a_1^3)b_1^{0.5}\log(\frac{1}{\epsilon_1}) + (a_2^2b_2+a_2^3)b_2^{0.5}\log(\frac{1}{\epsilon_2}))$.
\section{Simulation Results}\label{section.Numerical Results}
This section outlines the default simulation settings, verifies the DASHF algorithm's convergence, compares it with baselines for validation, and examines the effects of varying communication, computational resources, and cost weights on the SCR.
\subsection{Default settings}
We consider a network topology of 1000 m $\times$ 1000 m with 10 mobile users and 2 servers. The large-scale fading $h_{n,m}$ between the user $n$ and server $m$ is modeled as $128.1+37.6 \log_{10}d_{n,m}$, where $d_{n,m}$ denotes the  Euclidean distance between the user $n$ and server $m$. The small-scale fading is the Rayleigh fading. Gaussian noise power $\sigma^2$ is $-134$dBm. The total bandwidth for each server $b_{max}$ is 10 MHz. The maximum transmit power of mobile users $p_{max}^{(n)}$ is 0.2 W. The maximum transmit power of servers $p_{max}^{(m)}$ is 10 W. We assume the GPU resource utilization is $0.55$ for users and servers. The maximum GPU computation speed of mobile users $F_{max}^{(n)}$ is $19.58$ TFLOPs with four GTX 1080 GPUs and that of servers $F_{max}^{(m)}$ is $1372.8$ TFLOPs with eight A100 GPUs. The effective switched capacitance of mobile users and servers ($\kappa_n$ and $\kappa_m$) is $10^{-38}$. We refer to the adapter parameter sizes in \cite{zhang2023llama} and \cite{gao2023llama}. The adapter parameter sizes of mobile users are randomly selected from $[1.2, 14]$ M. To achieve this, pseudorandom values are generated, which follow a standard uniform distribution over the open interval $(0,1)$. These pseudorandom values are then scaled to the range of $[1.2, 14]$ M to determine the specific adapter parameter sizes for each mobile user. The token data sizes of users are randomly selected from $[10, 50]$ M bits. The parameters of delay and energy consumption ($\omega_t$ and $\omega_e$) are 0.5 and 0.005 to keep them in the same order. We set $\varpi_1$ and $\varpi_2$ as $\frac{10000}{\ln2}$ and $\frac{1}{3}$ to keep SCR big enough, respectively. we consider the ``float32'' method to represent the floating-point number and  $\omega_b$ is 32. User and server training epochs $e_n$ and $e_m$ are both one. The Mosek tool in Matlab is used to conduct the simulations.

% With respect to the parameter size of the adapter, it needs 1.2 M parameters in \cite{zhang2023llama}; it needs 14 M parameters in \cite{gao2023llama}. Therefore, we can set the parameter size ranging from 1.2 M to 14 M. If parameters are represented using 32-bit single-precision floating-point numbers, then the parameter data size will range from 38.4 Mbits to 448 Mbits.

% Consider the user token data size ranges from 10 Mbits to 50 Mbits.

% Consider the server has 8 40GB A100 GPUs and its peak FLOPs is $312\times 10^{12}$. Generally speaking, the GPU utilization is between 0.3 and 0.55. Therefore, the total available GPU FLOPs capacity is $8 \times 312\times 10^{12} \times 0.55$.

\subsection{Convergence of proposed Algorithms}
In this section, the convergence of our algorithms is evaluated using two network topologies: (10 users, 2 servers) and (20 users, 3 servers), with other settings at default. Convergence is assessed based on the primal objective value, which estimates the optimization problem's primal objective when solved using Mosek. Figures \ref{fig:Convergence_of_AO-Part1}, \ref{fig:Convergence_of_AO-Part2}, and \ref{fig:Convergence_of_AO} demonstrate the convergence of Algorithms AO-Part 1, AO-Part 2, and the DASHF Algorithm within 15, 9, and 9 iterations respectively, proving the DASHF Algorithm's effectiveness in finding a stationary point for Problem (\ref{prob1}).
% In this section, we evaluate the convergence of the proposed algorithms. We consider two network topologies with (10 users, 2 servers) and (20 users, 3 servers) and keep other settings as default. Primal objective value means an estimate for the primal objective value when using Mosek to solve the optimization problem. When the primal objective value converges, the algorithm converges to one stationary point. Fig. \ref{fig:Convergence_of_AO-Part1} plots the convergence of the Algorithm AO-Part 1, which converges within 15 iterations. Fig. \ref{fig:Convergence_of_AO-Part2} plots the convergence of the Algorithm AO-Part 2, which converges within 9 iterations. Fig. \ref{fig:Convergence_of_AO} plots the convergence of the DASHF Algorithm, which converges within 9 iterations. Thus, the proposed DASHF Algorithm is effective in finding one stationary point of the Problem (\ref{prob1}).
\subsection{Comparison with baselines}
In this section, we mainly consider four baselines to carry out the comparison experiments.
\begin{enumerate}
    \item \textbf{\underline{R}andom \underline{u}ser \underline{c}onnection with \underline{a}verage resource \underline{a}llocation (RUCAA)}. In the algorithm, each user is randomly assigned a server, which then evenly distributes communication and computational resources among its connected users.
    \item \textbf{\underline{G}reedy \underline{u}ser \underline{c}onnection with \underline{a}verage resource \underline{a}llocation (GUCAA)}. In this algorithm, each user chooses the server with the fewest underserved users, and the server equally allocates communication and computational resources to its connected users.
    \item \textbf{\underline{A}verage resource \underline{a}llocation with \underline{u}ser \underline{c}onnection \underline{o}ptimization (AAUCO)}. This algorithm assigns equal communication and computation resources from each Metaverse server to its connected users, utilizing \textbf{Algorithm \ref{section-AO-part1}} for user connection optimization.
    \item \textbf{\underline{G}reedy \underline{u}ser \underline{c}onnection with \underline{r}esource allocation \underline{o}ptimization (GUCRO)}. Each user chooses the Metaverse server with the fewest underserved users and \textbf{Algorithm \ref{section-AO-part2}} is used for resource optimization.
    \item \textbf{Proposed DASHF algorithm}. Joint optimization of user connection and resource allocation by utilizing the whole proposed DASHF algorothm.
\end{enumerate}
In Fig. \ref{Fig.AO performance}, the DASHF algorithm is compared against various baselines for resource consumption and SCR. RUCAA and GUCAA underperform due to lack of optimization, while GUCRO and AAUCO show improved results, validating the effectiveness of Algorithm AO Part 1 and 2. Notably, AAUCO's SCR outperforms GUCRO, indicating superior efficacy of user connection optimization over resource optimization in this scenario. The DASHF algorithm excels in time consumption, maintains low energy use (slightly more than AAUCO), and achieves the highest SCR, benefiting from the combined optimization of user connection and resource allocation.
% In Fig. \ref{Fig.AO performance}, we compare the resource consumption and SCR of the proposed DASHF Algorithms with other baselines. The performances of RUCAA and GUCAA are worse since no optimization is utilized. GUCRO and AAUCO have better performances than GUCAA, which confirms the effectiveness of the proposed Algorithm AO-Part 1 and Part 2. Furthermore, the SCR of AAUCO is higher than that of GUCRO, which shows that user connection optimization is more effective than resource optimization in this case. The time consumption of the proposed DASHF algorithm is the lowest of these five methods and the energy consumption is also low (just higher than AAUCO), and the SCR is the highest one. This results from the benefits of joint optimization of user connection and resource allocation.
\subsection{SCR versus the total bandwidth}
We consider the total bandwidth from 10 MHz to 100 MHz to test the SCR under different total bandwidths. Other parameters are fixed as default settings. Fig. \ref{Fig.Comparison_of_b_max} reveals distinct algorithmic performance trends, with the proposed DASHF method consistently outperforming GUCRO, AAUCO, RUCAA, and GUCAA in terms of the SCR. Notably, optimization algorithms (GUCRO and AAUCO) demonstrate superior or close performance compared to non-optimization algorithms (RUCAA and GUCAA). AAUCO employs user connection optimization strategies and performs better than RUCAA, GUCRO, and GUCAA. 
\subsection{Impact of cost weights on SCR}
Fig. \ref{Fig.Comparison_of_omega} displays how different ($\omega_t$, $\omega_e$) combinations affect the trade-off between delay-energy optimization and trust score. Shifting these values alters outcomes: prioritizing energy efficiency (e.g., ($\omega_t$, $\omega_e$) = (0.1, 0.009)) leads to lower energy but higher delay, yielding moderate SCR. In contrast, balanced settings (e.g., ($\omega_t$, $\omega_e$) = (0.5, 0.005)) achieve lower delay and slightly higher energy, resulting in high SCR. This underscores the need to finely tune these parameters to suit specific application requirements and manage the delay-energy versus service score balance.
% Fig. \ref{Fig.Comparison_of_omega} featuring various combinations of ($\omega_t$, $\omega_e$) that signifies the trade-off between delay-energy optimization and trust score. As ($\omega_t$, $\omega_e$) values shift, emphasizing either delay or energy, distinct performance outcomes are evident. For instance, when prioritizing energy efficiency (e.g., ($\omega_t$, $\omega_e$) = (0.1, 0.009)), the system achieves lower energy consumption but at the expense of higher delay, resulting in a moderate SCR value. Conversely, balanced settings (e.g., ($\omega_t$, $\omega_e$) = (0.5, 0.005)) lead to lower delay and slightly higher energy consumption, yielding a high SCR value. These results highlight the optimization's sensitivity to parameters and the need to balance ($\omega_t$, $\omega_e$) with application needs and delay-energy versus service score trade-offs.
\section{Conclusion}\label{section.Conclusion}
This study explores efficient LLM service delivery within wireless communication system constraints. We introduce a collaborative training approach between mobile users and servers, tailored to overcome resource limitations. Our strategy optimizes resource use and maintains robust LLM performance, assigning initial layer training to users and subsequent layers to servers. The efficiency and resource optimization are gauged by the SCR, with the DASHF algorithm being key to our approach. Ultimately, mobile edge computing is poised to enable broad, efficient access to advanced LLM services, harmonizing computational constraints with growing application needs.
% In this investigation into LLMs and MEC, we've delved into the intricacies of ensuring efficient LLM service delivery amidst the constraints of wireless communication systems. With their vast linguistic and computational capabilities, the promise of LLMs is now being actualized in real-world applications. Our contributions, as presented in this paper, lay the foundation for seamless collaborative training between mobile users and servers, addressing the challenges of limited computational and communication resources. We optimize resource utilization and ensure robust LLM performance by implementing a framework where initial layers are trained by users and subsequent layers by servers. In this context, the SCR measures collaboration efficiency and resource optimization. The DASHF algorithm, central to our methodology, solidifies these efforts. In conclusion, mobile edge computing is expected to provide widespread, efficient access to advanced LLM services, balancing computational limits with rising application demands.
\bibliographystyle{IEEEtran}
\bibliography{ref}

\begin{thebibliography}{10}
\providecommand{\url}[1]{#1}
\csname url@samestyle\endcsname
\providecommand{\newblock}{\relax}
\providecommand{\bibinfo}[2]{#2}
\providecommand{\BIBentrySTDinterwordspacing}{\spaceskip=0pt\relax}
\providecommand{\BIBentryALTinterwordstretchfactor}{4}
\providecommand{\BIBentryALTinterwordspacing}{\spaceskip=\fontdimen2\font plus
\BIBentryALTinterwordstretchfactor\fontdimen3\font minus \fontdimen4\font\relax}
\providecommand{\BIBforeignlanguage}[2]{{%
\expandafter\ifx\csname l@#1\endcsname\relax
\typeout{** WARNING: IEEEtran.bst: No hyphenation pattern has been}%
\typeout{** loaded for the language `#1'. Using the pattern for}%
\typeout{** the default language instead.}%
\else
\language=\csname l@#1\endcsname
\fi
#2}}
\providecommand{\BIBdecl}{\relax}
\BIBdecl

\bibitem{zhao2023human}
J.~Zhao, L.~Qian, and W.~Yu, ``Human-centric resource allocation in the {Metaverse} over wireless communications,'' \emph{IEEE Journal on Selected Areas in Communications (JSAC)}, vol.~42, no.~3, pp. 514--537, 2024.

\bibitem{gao2023llama}
P.~Gao, J.~Han, R.~Zhang, Z.~Lin, S.~Geng, A.~Zhou, W.~Zhang, P.~Lu, C.~He, X.~Yue \emph{et~al.}, ``Llama-adapter v2: Parameter-efficient visual instruction model,'' \emph{arXiv preprint arXiv:2304.15010}, 2023.

\bibitem{zhang2023llama}
R.~Zhang, J.~Han, A.~Zhou, X.~Hu, S.~Yan, P.~Lu, H.~Li, P.~Gao, and Y.~Qiao, ``Llama-adapter: Efficient fine-tuning of language models with zero-init attention,'' \emph{arXiv preprint arXiv:2303.16199}, 2023.

\bibitem{dong2023lambo}
L.~Dong, F.~Jiang, Y.~Peng, K.~Wang, K.~Yang, C.~Pan, and R.~Schober, ``Lambo: Large language model empowered edge intelligence,'' \emph{arXiv preprint arXiv:2308.15078}, 2023.

\bibitem{shen2023large}
Y.~Shen, J.~Shao, X.~Zhang, Z.~Lin, H.~Pan, D.~Li, J.~Zhang, and K.~B. Letaief, ``Large language models empowered autonomous edge ai for connected intelligence,'' \emph{arXiv preprint arXiv:2307.02779}, 2023.

\bibitem{hong2023intelligence}
Z.~Hong, X.~Qiu, J.~Lin, W.~Chen, Y.~Yu, H.~Wang, S.~Guo, and W.~Gao, ``Intelligence-endogenous management platform for computing and network convergence,'' \emph{IEEE Network}, 2023.

\bibitem{guo2023promptfl}
T.~Guo, S.~Guo, J.~Wang, X.~Tang, and W.~Xu, ``Promptfl: Let federated participants cooperatively learn prompts instead of models-federated learning in age of foundation model,'' \emph{IEEE Transactions on Mobile Computing}, 2023.

\bibitem{lai2023resource}
B.~Lai, J.~Wen, J.~Kang, H.~Du, J.~Nie, C.~Yi, D.~I. Kim, and S.~Xie, ``Resource-efficient generative mobile edge networks in 6g era: Fundamentals, framework and case study,'' \emph{arXiv preprint arXiv:2312.12063}, 2023.

\bibitem{zeng2021energy}
Q.~Zeng, Y.~Du, K.~Huang, and K.~K. Leung, ``Energy-efficient resource management for federated edge learning with {CPU-GPU} heterogeneous computing,'' \emph{IEEE Transactions on Wireless Communications}, vol.~20, no.~12, pp. 7947--7962, 2021.

\bibitem{yang2015incentive}
D.~Yang, G.~Xue, X.~Fang, and J.~Tang, ``Incentive mechanisms for crowdsensing: Crowdsourcing with smartphones,'' \emph{IEEE/ACM Transactions on Networking}, vol.~24, no.~3, pp. 1732--1744, 2015.

\bibitem{dinkelbach1967nonlinear}
W.~Dinkelbach, ``On nonlinear fractional programming,'' \emph{Management Science}, vol.~13, no.~7, pp. 492--498, 1967.

\bibitem{qian2023user}
L.~Qian and J.~Zhao, ``User association and resource allocation in large language model based mobile edge computing system over wireless communications,'' \emph{arXiv preprint arXiv:2310.17872}, 2023.

\bibitem{dai2018joint}
Y.~Dai, D.~Xu, S.~Maharjan, and Y.~Zhang, ``Joint computation offloading and user association in multi-task mobile edge computing,'' \emph{IEEE Transactions on Vehicular Technology}, vol.~67, no.~12, pp. 12\,313--12\,325, 2018.

\end{thebibliography}

% \section*{Authors}

% \begin{description}
%     \item[Qian Liangxin] (Graduate Student Member, IEEE) received bachelor's and master's degrees in communication engineering from the University of Electronic Science and Technology of China, Chengdu, China, in 2019 and 2022, respectively. He is currently working toward his Ph.D. at the School of Computer Science and Engineering, Nanyang Technological University, Singapore. His research interests include Metaverse, mobile edge computing, and communication theory.
% \end{description}
\end{document}